\documentclass[11pt]{article}

\usepackage{amsthm}
\usepackage{graphicx} % support the \includegraphics command and options
\usepackage{array} % for better arrays (eg matrices) in maths

\usepackage[bb=boondox]{mathalpha}

\usepackage{amsmath, amssymb, amsfonts, verbatim, soul}
\usepackage{hyphenat,epsfig,subcaption,multirow}
\usepackage[font=small,labelfont=bf]{caption}
\usepackage{nicefrac}
\usepackage{paralist}

\usepackage[usenames,dvipsnames]{xcolor}
\usepackage[ruled]{algorithm2e}

\DeclareFontFamily{U}{mathx}{\hyphenchar\font45}
\DeclareFontShape{U}{mathx}{m}{n}{
      <5> <6> <7> <8> <9> <10>
      <10.95> <12> <14.4> <17.28> <20.74> <24.88>
      mathx10
      }{}
\DeclareSymbolFont{mathx}{U}{mathx}{m}{n}
\DeclareMathSymbol{\bigtimes}{1}{mathx}{"91}

\usepackage{tcolorbox}
\tcbuselibrary{skins,breakable}
\tcbset{enhanced jigsaw}

\usepackage[normalem]{ulem}
\usepackage[compact]{titlesec}

\definecolor{DarkRed}{rgb}{0.5,0.1,0.1}
\definecolor{DarkBlue}{rgb}{0.1,0.1,0.5}

\usepackage{nameref}
\definecolor{ForestGreen}{rgb}{0.1333,0.5451,0.1333}
\definecolor{Red}{rgb}{0.9,0,0}
\usepackage[linktocpage=true,
	pagebackref=true,colorlinks,
	linkcolor=DarkRed,citecolor=ForestGreen,
	bookmarks,bookmarksopen,bookmarksnumbered]
	{hyperref}
\usepackage[noabbrev,nameinlink]{cleveref}
\crefname{property}{property}{Property}
\creflabelformat{property}{(#1)#2#3}
\crefname{equation}{eq}{Eq}
\creflabelformat{equation}{(#1)#2#3}
\usepackage{bm}
\usepackage{url}
\usepackage{xspace}
\usepackage[mathscr]{euscript}

\usepackage{tikz}
\usetikzlibrary{arrows}
\usetikzlibrary{arrows.meta}
\usetikzlibrary{shapes}
\usetikzlibrary{backgrounds}
\usetikzlibrary{positioning}
\usetikzlibrary{decorations.markings}
\usetikzlibrary{decorations.pathreplacing} % needed for drawing braces
\usetikzlibrary{patterns}
\usetikzlibrary{calc}
\usetikzlibrary{fit}
\usetikzlibrary{decorations}

\usepackage[framemethod=TikZ]{mdframed}

\usepackage[noend]{algpseudocode}
\makeatletter
\def\BState{\State\hskip-\ALG@thistlm}
\makeatother

\usepackage{cite}
\usepackage{enumitem}
\setlist[itemize]{leftmargin=20pt}
\setlist[enumerate]{leftmargin=20pt}

\usepackage[margin=1in]{geometry}

\usepackage{thmtools}
\usepackage{thm-restate}

\newtheorem{theorem}{Theorem}

\newtheorem{lemma}{Lemma}[section]

\newtheorem{corollary}[lemma]{Corollary}
\newtheorem{claim}[lemma]{Claim}
\newtheorem{fact}[lemma]{Fact}

\newtheorem{definition}[lemma]{Definition}

\newtheorem*{claim*}{Claim}
\newtheorem*{assumption*}{Assumption}
\newtheorem*{proposition*}{Proposition}
\newtheorem*{lemma*}{Lemma}

\newtheorem*{observation*}{Observation}

\newtheorem*{theorem*}{Theorem}

\crefname{lemma}{Lemma}{Lemmas}
\crefname{claim}{claim}{claims}
\crefname{property}{Property}{Properties}
\crefname{invariant}{Invariant}{Invariants}

\newtheorem{mdresult}[theorem]{Theorem}
\newenvironment{Theorem}{\begin{mdframed}[backgroundcolor=lightgray!40,topline=false,rightline=false,leftline=false,bottomline=false,innertopmargin=2pt]\begin{mdresult}}{\end{mdresult}\end{mdframed}}

\newtheorem{remark}[lemma]{Remark}

\theoremstyle{definition}

\newenvironment{ourbox}{\begin{mdframed}[hidealllines=false,innerleftmargin=10pt,backgroundcolor=white!10,innertopmargin=10pt,innerbottommargin=5pt,roundcorner=10pt]}{\end{mdframed}}

\newenvironment{cbox}{\begin{mdframed}[backgroundcolor=ForestGreen!10,topline=false,bottomline=false, innerbottommargin=5pt,innertopmargin=5pt]}{\end{mdframed}}

\newtheorem{Definition}[lemma]{Definition}
\renewenvironment{definition}{\begin{cbox}\begin{Definition}}{\end{Definition}\end{cbox}}

\newtheorem{mdalgorithm}{Algorithm}

\allowdisplaybreaks

\renewcommand{\qed}{\nobreak \ifvmode \relax \else
      \ifdim\lastskip<1.5em \hskip-\lastskip
      \hskip1.5em plus0em minus0.5em \fi \nobreak
      \vrule height0.75em width0.5em depth0.25em\fi}

\newcommand{\Qed}[1]{\rlap{\qed$_{\textnormal{~~\Cref{#1}}}$}}

\newcommand*\samethanks[1][\value{footnote}]{\footnotemark[#1]}

\renewcommand{\leq}{\leqslant}
\renewcommand{\geq}{\geqslant}

\newcommand{\Ot}{\ensuremath{\widetilde{O}}}

\newcommand{\Paren}[1]{\Big(#1\Big)}

\newcommand{\paren}[1]{\ensuremath{\left(#1\right)}\xspace}
\newcommand{\card}[1]{\left\vert{#1}\right\vert}

\newcommand{\set}[1]{\ensuremath{\left\{ #1 \right\}}}
\newcommand{\poly}{\mbox{\rm poly}}

\DeclareMathOperator*{\Prob}{\ensuremath{\textnormal{Pr}}}
\renewcommand{\Pr}{\Prob}

\newenvironment{tbox}{\begin{tcolorbox}[
		enlarge top by=5pt,
		enlarge bottom by=5pt,
		 breakable,
		 boxsep=0pt,
                  left=4pt,
                  right=4pt,
                  top=10pt,
                  arc=0pt,
                  boxrule=1pt,toprule=1pt,
                  colback=white
                  ]%%
	}
{\end{tcolorbox}}

\newcommand{\II}{\ensuremath{\mathbf{I}}}

\newcommand{\mireal}[1][]{
  \ifx\relax#1\relax%
    \II(\mione \,; \mitwo)%
  \else%
    \II(\mione \,; \mitwo\mid #1)%
  \fi
}

\newcommand{\conc}{\circ}

\newcommand{\bG}{\mathbb{G}}

\newcommand{\bE}{\mathbb{E}}
\newcommand{\bL}{\mathbb{L}}
\newcommand{\bR}{\mathbb{R}}

\newcommand{\be}{\mathbb{e}}

\newcommand{\bEE}[1]{\ensuremath{\bE}^{(#1)}}

\newcommand{\bX}{\mathbb{X}}

\newcommand{\bM}{\mathbb{M}}

\newcommand{\indextext}[1]{\ensuremath{\textnormal{\texttt{#1}}}\xspace}
\newcommand{\pstart}{\indextext{start}}
\newcommand{\pdelay}{\indextext{delay}}
\newcommand{\pwalk}{\indextext{walk}}

\newcommand{\barW}{\overline{W}}

\title{Semi-Streaming Matching in a Single Pass II: \\ Greedy is Optimal} 
\author{Sepehr Assadi\footnote{(sepehr@assadi.info, max.jiang@uwaterloo.ca, mxiang@uwaterloo.ca) School of Computer Science, University of Waterloo. Supported in part by an NSERC
Discovery Grant, a Faculty of Math Research Chair grant, and an Undergraduate Research Fellowship (URF) from University of Waterloo. \newline  Part of this work was done while Mars Xiang was
at The Interaction Company of California. \smallskip} 
\and Max Jiang\samethanks  \and Mars Xiang\samethanks  
}

\date{}

\begin{document}
\maketitle

% !TeX root = main.tex 
%!TEX root = main.tex

\begin{abstract}

\medskip

We prove that no single-pass semi-streaming algorithm (deterministic or randomized) can achieve a better-than-half approximation 
to the maximum matching problem. This implies the optimality of the naive greedy algorithm, answering an outstanding open question in the graph streaming literature since
the introduction of the model over two decades ago.

\medskip

Our proof follows the ``blueprint framework'' introduced previously by the authors, which reduced proving lower bounds for semi-streaming matching
to constructing certain combinatorial objects called blueprints. 
We present an optimal construction of blueprints that when used in this framework implies our semi-streaming matching lower bound. 

\medskip

Our results also imply that the optimal competitive ratio of online matching with preemption is half, again matching the naive greedy algorithm, settling this open question as well. 

\medskip

\end{abstract}

\pagenumbering{roman}
\clearpage

%%\vspace{-15pt}

\setcounter{tocdepth}{2}
\tableofcontents
\clearpage

\pagenumbering{arabic}
\setcounter{page}{1}

\clearpage

% !TeX root = main.tex 
%!TEX root = main.tex

\newcommand{\bLstar}{\bL^{\star}}
\newcommand{\bRstar}{\bR^{\star}}

\section{Introduction}\label{sec:intro} 

In the semi-streaming model for graph problems, formalized by~\cite{FeigenbaumKMSZ05}, we have an $n$-vertex graph $G=(V,E)$ whose edges are given
to the algorithm in some arbitrarily ordered stream; the algorithm makes one or a few passes over the stream, uses $\Ot(n):= O(n \cdot \poly\!\log{(n)})$ memory, and at the end of the stream, 
outputs a solution to the problem at hand. In this paper, we study \textbf{single-pass} semi-streaming algorithms for the \textbf{maximum matching} problem. 

Alongside the introduction of the model,~\cite{FeigenbaumKMSZ05} observed a simple $(1/2)$-approximation for this problem: greedily maintain a maximal matching of the arriving edges. 
The two decades since then have witnessed a high throughput of results on the semi-streaming matching problem, and yet, this simple algorithm remains the state of the art. A handful of 
impossibility results have also been proven over the years, ruling out approximation ratios of $(2/3)$~\cite{GoelKK12}, $(1-1/e) \sim 0.632$~\cite{Kapralov13}, $(1+\ln{(2)})^{-1} \sim 0.590$~\cite{Kapralov21}, 
and most recently $(8-2\sqrt{10})/3 \sim 0.558$~\cite{AssadiJX26}. 

In light of this state of affairs, whether the half-approximation of the greedy algorithm can be improved has been an outstanding open question since the introduction of the model~\cite{FeigenbaumKMSZ05,KonradMM12,sublinear_open_60,McGregor14,KaleT17,KonradN21}. 
Over the years, this question has been commonly referred to as ``major''~\cite{Tirodkar18,GamlathKMS19,GargKRS20,FischerMU22}, ``central''~\cite{AssadiBKL23,DerakhshanGR25}, ``key''~\cite{AlexandruDKN23}, ``main''~\cite{Bernstein24}, one of the ``most longstanding''~\cite{Behnezhad21,AssadiB21,DerakhshanGR25} or ``most appealing''~\cite{EsfandiariHM16,CormodeJMM17}, ``perhaps the most well-known''~\cite{Gamlath21}, and ``the most basic''~\cite{FeldmanS22} of open questions in the area. It has also been more affectionally called ``vexing''~\cite{ChakrabartiK14},``baffling''~\cite{KaleT17}, ``most tantalizing''~\cite{Wajc21}, ``notorious''~\cite{FeldmanNSZ22}, and ``elusive''~\cite{naidu2024streaming}, among others.

We fully resolve this question in our work. 

\begin{Theorem}\label{thm:main}
	There is no single-pass semi-streaming algorithm for the maximum bipartite matching problem that can achieve any constant approximation ratio strictly better than half with constant probability. 
\end{Theorem}

\Cref{thm:main} settles the complexity of the semi-streaming matching problem: after all, the naive greedy algorithm turns out to be optimal for this fundamental problem. 

Our proof of~\Cref{thm:main} is based on the \emph{``blueprint framework''} of the recent work of the authors in~\cite{AssadiJX26}. We describe this framework next and put our result in this context.

\subsection{Blueprint Framework of~\cite{AssadiJX26}} 

In~\cite{AssadiJX26}, we reduced proving lower bounds for semi-streaming matching to constructing ``constant-size'' (relative to stream length) combinatorial
objects called \emph{blueprints}. We now define blueprints and this framework\footnote{We directly define what~\cite{AssadiJX26} calls a \emph{simple} blueprint as the blueprint itself, 
as it is a lot easier to describe and we do not need the extra power of the more general definition.}
(following~\cite{AssadiJX26}, we use `blackboard bold' letters for blueprint notation).

\begin{definition}\label{def:blueprint}
    For integers $P,C \geq 1$, we define a \textbf{blueprint} with parameters $P,C$ as any bipartite graph $\bG=(\bL,\bR,\bE)$ such that: 
    \begin{itemize}
    \item Vertices in $\bL$ and $\bR$ are identified by tuples in $[C]^P$; specifically, for any $x \in [C]^P$ we have a vertex $\bL(x) \in \bL$ and a vertex $\bR(x) \in \bR$. 
    \item Edges of $\bG$ are partitioned into $P$ groups:
    $
      \bE = \bEE{1} \sqcup \cdots \sqcup \bEE{P}.
    $
    \end{itemize}
\end{definition}
\Cref{def:blueprint} simply gives the ``skeleton'' of a blueprint; what we really need from a blueprint is to satisfy the following two constraints (see~\Cref{sec:prelim} for a quick review of notation).  

\begin{definition}\label{def:proper-blueprint}
	We say a blueprint $\bG=(\bL,\bR,\bE)$ with parameters $(P,C)$ is \textbf{proper} if it satisfies the following constraints:
    \begin{itemize}
        \item \textnormal{\textbf{Matching constraint:}} The edges in $\bE$ form a single matching in $\bG$. 

        \item \textnormal{\textbf{Ban constraints:}} For any $x,y \in [C]^P$ and vertices $\bL(x)$ and $\bR(y)$ incident on the same edge $\be=(\bL(x),\bR(y)) \in \bEE{p}$ for some $p \in [P]$:
	\begin{quote}
		for any $z \in [C]^{P-p+1}$, at least one of vertices $\bL(x_{<p} \conc z)$ or $\bR(y_{<p} \conc z)$ has no edges in $\bG$; we
		say these two vertices are \textbf{banned together} by the edge $\be$.
	\end{quote}
    \end{itemize}
\end{definition}
See~\Cref{fig:ban} for an example of ban constraints.

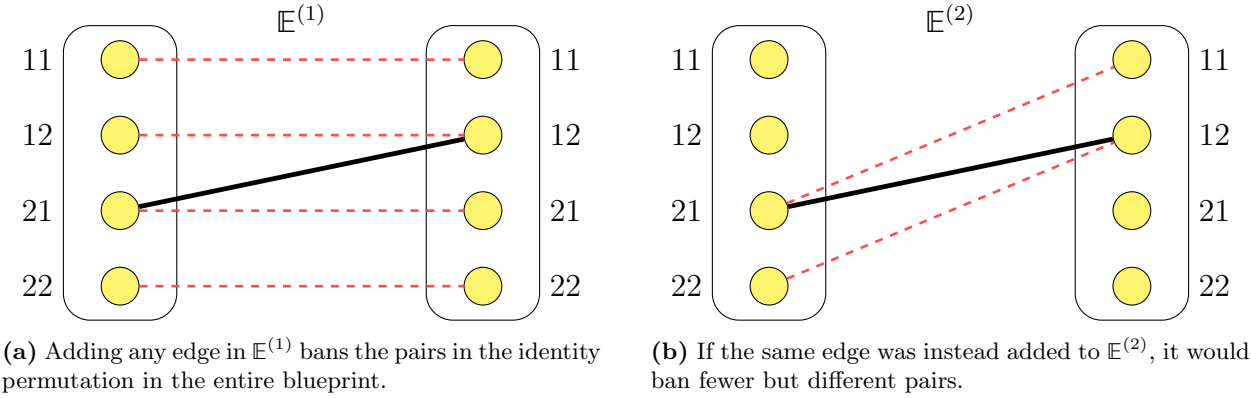
\begin{figure}[h!]
\centering

\begin{subfigure}[t]{0.48\textwidth}
\centering
\begin{tikzpicture}[scale=1, every node/.style={font=\large}]
    % styles
    \tikzstyle{vtx}=[circle, draw=black, fill=yellow!70, minimum size=5mm, inner sep=0pt]
    \tikzstyle{dashededge}=[red!70, dashed, line width=1.0pt]
    \tikzstyle{solidedge}=[black, line width=1.8pt]

    % coordinates
    \def\lx{0}
    \def\rx{4.8}
    \def\yA{3}
    \def\yB{2}
    \def\yC{1}
    \def\yD{0}

    % enclosing boxes
    \draw[rounded corners=10pt] (-0.75,-0.45) rectangle (0.75,3.45);
    \draw[rounded corners=10pt] (4.05,-0.45) rectangle (5.55,3.45);

    % left vertices
    \node[vtx] (L11) at (\lx,\yA) {};
    \node[vtx] (L12) at (\lx,\yB) {};
    \node[vtx] (L21) at (\lx,\yC) {};
    \node[vtx] (L22) at (\lx,\yD) {};

    % right vertices
    \node[vtx] (R11) at (\rx,\yA) {};
    \node[vtx] (R12) at (\rx,\yB) {};
    \node[vtx] (R21) at (\rx,\yC) {};
    \node[vtx] (R22) at (\rx,\yD) {};

    % labels
   % labels: placed 0pt away from the enclosing boxes
    \node[anchor=east] at ([xshift=0pt]-0.75,\yA) {$11$};
    \node[anchor=east] at ([xshift=0pt]-0.75,\yB) {$12$};
    \node[anchor=east] at ([xshift=0pt]-0.75,\yC) {$21$};
    \node[anchor=east] at ([xshift=0pt]-0.75,\yD) {$22$};

    \node[anchor=west] at ([xshift=0pt]5.55,\yA) {$11$};
    \node[anchor=west] at ([xshift=0pt]5.55,\yB) {$12$};
    \node[anchor=west] at ([xshift=0pt]5.55,\yC) {$21$};
    \node[anchor=west] at ([xshift=0pt]5.55,\yD) {$22$};

    % title
    \node at (2.4,3.55) {$\mathbb{E}^{(1)}$};

    % dashed edges
    \draw[dashededge] (L11) -- (R11);
    \draw[dashededge] (L12) -- (R12);
    \draw[dashededge] (L21) -- (R21);
    \draw[dashededge] (L22) -- (R22);

    % solid edge
    \draw[solidedge] (L21) -- (R12);
\end{tikzpicture}
\caption{Adding any edge in $\bEE{1}$ bans the pairs in the identity permutation in the entire blueprint.}
\end{subfigure}
\hfill
\begin{subfigure}[t]{0.48\textwidth}
\centering
\begin{tikzpicture}[scale=1, every node/.style={font=\large}]
    % styles
    \tikzstyle{vtx}=[circle, draw=black, fill=yellow!70, minimum size=5mm, inner sep=0pt]
    \tikzstyle{dashededge}=[red!70, dashed, line width=1.0pt]
    \tikzstyle{solidedge}=[black, line width=1.8pt]

    % coordinates
    \def\lx{0}
    \def\rx{4.8}
    \def\yA{3}
    \def\yB{2}
    \def\yC{1}
    \def\yD{0}

    % enclosing boxes
    \draw[rounded corners=10pt] (-0.75,-0.45) rectangle (0.75,3.45);
    \draw[rounded corners=10pt] (4.05,-0.45) rectangle (5.55,3.45);

    % left vertices
    \node[vtx] (L11) at (\lx,\yA) {};
    \node[vtx] (L12) at (\lx,\yB) {};
    \node[vtx] (L21) at (\lx,\yC) {};
    \node[vtx] (L22) at (\lx,\yD) {};

    % right vertices
    \node[vtx] (R11) at (\rx,\yA) {};
    \node[vtx] (R12) at (\rx,\yB) {};
    \node[vtx] (R21) at (\rx,\yC) {};
    \node[vtx] (R22) at (\rx,\yD) {};

    % labels
   % labels: placed 0pt away from the enclosing boxes
    \node[anchor=east] at ([xshift=0pt]-0.75,\yA) {$11$};
    \node[anchor=east] at ([xshift=0pt]-0.75,\yB) {$12$};
    \node[anchor=east] at ([xshift=0pt]-0.75,\yC) {$21$};
    \node[anchor=east] at ([xshift=0pt]-0.75,\yD) {$22$};

    \node[anchor=west] at ([xshift=0pt]5.55,\yA) {$11$};
    \node[anchor=west] at ([xshift=0pt]5.55,\yB) {$12$};
    \node[anchor=west] at ([xshift=0pt]5.55,\yC) {$21$};
    \node[anchor=west] at ([xshift=0pt]5.55,\yD) {$22$};

    % title
    \node at (2.4,3.55) {$\mathbb{E}^{(2)}$};

    % dashed edges
    \draw[dashededge] (L21) -- (R11);
    \draw[dashededge] (L22) -- (R12);

    % solid edge
    \draw[solidedge] (L21) -- (R12);
\end{tikzpicture}
\caption{If the same edge was instead added to $\bEE{2}$, it would ban fewer but different pairs.}
\end{subfigure}

\caption{An illustration of ban constraints in a blueprint with $P=C=2$ with labels of vertices drawn next to them. Solid (black) edges show two possible edges in two different blueprints and dashed (red) edges show banned pairs resulting
from inclusion of each of these edges separately.}
\label{fig:ban}
\end{figure}

Finally, the following quantity determines how ``good'' a blueprint is for our purpose.
\begin{definition}\label{def:blueprint-value-approx}
	Define the \textnormal{\textbf{value}} of a blueprint $\bG = (\bL,\bR,\bE)$ as:
 	 \[
    		value(\bG) := \frac{\card{\bE}}{\card{\bL}} = \frac{\card{\bE}}{\card{\bR}} = \frac{\card{\bE}}{C^P}. 
    	\]
\end{definition}

We refer the reader to~\cite{AssadiJX26} for several examples of blueprints and the intuition (and the formal proof) behind the connection between blueprints and semi-streaming matching. For us, we need the following 
theorem that captures the blueprint framework put forward in~\cite{AssadiJX26}\footnote{In light of what we said earlier about \emph{simple} blueprints in~\cite{AssadiJX26}, it is worth 
emphasizing that~\cite{AssadiJX26} proves this theorem for simple blueprints and then shows simple and not-simple blueprints are equivalent; given our construction is already a simple blueprint, we 
do not need this equivalence and defined simple blueprints as blueprints.}.

\begin{theorem}[\!\!{\cite{AssadiJX26}}]\label{thm:framework}
	Suppose there is a proper blueprint $\bG$ and define $\alpha := \frac{2-2 \cdot value(\bG)}{2-value(\bG)}$. 
	Then, for any fixed $\delta > 0$, and any sufficiently large $n$, as a function of $\bG$ and $\delta$, the following holds. 
	Any single-pass semi-streaming algorithm for maximum bipartite matching on $n$-vertex graphs cannot achieve
	an $(\alpha+\delta)$-approximation with probability more than $1-\delta$. 
\end{theorem}

In~\cite{AssadiJX26}, we constructed blueprints with values converging to $(5-\sqrt{10})/3$ which, plugged in~\Cref{thm:framework}, proves the impossibility of better than $(8-2\sqrt{10})/3 \sim 0.558$ approximation
for the semi-streaming matching problem. The limit of the framework of~\cite{AssadiJX26} however is determined by the following quantity
\begin{align}
\sup\set{value(\bG):\text{$\bG$ a proper blueprint}},  \label{eq:sup}
\end{align}
which was not well understood in~\cite{AssadiJX26}. 

\paragraph{Our blueprint construction.} 
We determine the optimal value of~\Cref{eq:sup} in this work, proving the following, which is the main \emph{technical} result of our work. 
\begin{theorem}\label{thm:blueprint}
	There is a sequence of proper blueprints $(\bG_m)_{m \geq 1}$ with parameter $P_m \to \infty$ and 
	\[
		\lim_{m \to \infty} value(\bG_m) = 2/3.
	\] 
\end{theorem}

Plugging in the blueprints in~\Cref{thm:blueprint} in~\Cref{thm:framework} then immediately proves~\Cref{thm:main}. Our proof of~\Cref{thm:blueprint} is  constructive and we 
give an explicit description of the blueprints in the family (even though this is not required to apply~\Cref{thm:framework}). We also note that as pointed out in~\cite{AssadiJX26}, 
one does need $P \to \infty$ for the value of the blueprint to be able to reach $2/3$.

As an aside, we should note that, technically speaking, one also needs to prove an upper bound on~\Cref{eq:sup} to fully determine the value of~\Cref{eq:sup} (although this is not needed for proving~\Cref{thm:main}). However, a $2/3$ upper bound follows 
immediately from~\Cref{thm:framework} and the greedy algorithm for the semi-streaming matching problem. We give a different and direct proof of this upper bound in~\Cref{app:2-3-upper} as it provides further insight into the structure of blueprints.

\subsection{Our Blueprint Construction at a High Level}\label{sec:ideas} 

Blueprints in~\Cref{thm:blueprint} are created by using elementary properties of random walks on the integer line. The general idea of the construction is as follows (as our blueprints 
are \emph{entirely disjoint} from~\cite{AssadiJX26} technique-wise, we will not review those blueprints and refer the reader to~\cite{AssadiJX26}). 

\paragraph{Vertices and random walks.} Let $m \geq 1$ be an integer and $\bG = (\bL,\bR,\bE)$ be a blueprint with parameters $P$ and $C$ that depend on $m$. 
To any label $x \in [C]^P$, we assign a walk $W(x)$ on $\set{0,\ldots,2m}$ that starts from some state in $[2m-1]$ and in each step moves to either a one larger or one smaller 
state and terminates when it reaches $0$ or $2m$. Define the \emph{reflection} walk $\barW(x)$ that starts at the same state as $W(x)$ but in every step, takes the opposite direction as $W(x)$. 
This assignment is done so that for $x$ chosen uniformly from $[C]^P$, $W(x)$ (also $\barW(x)$) becomes a random walk whose starting state is sampled from some distribution $\mu$ and 
each of its step is chosen uniformly. This is done by using the first index of $[P]$ 
to simulate sampling from $\mu$ (by partitioning $[C]$ into intervals of lengths proportional to $\mu(i)$-values) and remaining indices to simulate the random walk (by partitioning $[C]$ into two equal halves); see~\Cref{fig:W(x)} for a quick illustration (we have
skipped one other coordinate in $[P]$ called ``delay'' and will explain it later). 

\paragraph{Edges between $0$-reaching random walks.} Let $\bLstar$ (resp. $\bRstar$) be vertices $\bL(x)$ (resp. $\bR(x)$) whose random walk $W(x)$ (resp. $\barW(x)$) reaches $0$. Moreover, for any $j \in [2m-1]$ and time $t \in [P-1]$, let $\bLstar_{j,t}$ denote the following subset of $\bLstar$: 
\begin{itemize}
	\item All $\bL(x) \in \bL$ such that $W(x)$ reaches $0$, $W(x)$ has $j$ as the \emph{maximum} state it ever visits, and, it has visited $j$ for the \emph{first} time at the $t$-th step of the walk. 
\end{itemize} 
We define $\bRstar_{j,t}$ analogously using the reflection walk $\barW$. In the blueprint, we add a \emph{matching} between $\bLstar_{j,t}$ and $\bRstar_{2m-j,t}$ in $\bEE{t+1}$ with size equal to the smaller of the two. 

\paragraph{Properness of the blueprint.} We need to ensure the inserted edges form a matching and that the ban constraints are satisfied. The former is easy to verify as $\bLstar_{j,t},\bRstar_{j,t}$ sets are disjoint for different $j,t$ and each is incident on a matching. We now prove the ban constraint.

Consider any edge between $\bL(x) \in \bLstar_{j,t}$ and $\bR(y) \in \bRstar_{2m-j,t}$ in some $\bEE{t+1}$ for $j \in [2m-1]$ and $t \in [P-1]$ (all edges are of this form). The ban constraint requires that for any $z \in [C]^{P-t}$, at most one of $\bL(x_{\leq t} \circ z)$ or $\bR(y_{\leq t} \circ z)$ is matched. Fixing $x_{\leq t}$ and $y_{\leq t}$ ``brings'' any walk of $W(x_{\leq t} \circ z)$ and $\barW(y_{\leq t} \circ z)$ 
to the states $j$ and $2m-j$ at their $t$-th step (by the definitions of $\bLstar_{j,t},\bRstar_{2m-j,t}$). Now, suppose $\bL(x_{\leq t} \circ z)$ is matched. Then, it means that the remainder of the walk, which with a slight abuse of notation we denote by $W(z)$, 
will move from the state $j$, uses $z$, and reaches $0$ (so $\bL(x_{\leq t} \circ z)$ becomes matched). But then, the reflection walk $\barW(z)$ will move in the ``opposite'' direction: it moves from $2m-j$, uses $z$, and thus reaches $2m$ instead; this means 
that $\bR(y_{\leq t} \circ z)$ will not be matched, as desired. This is the entire proof of the ban constraint (see also~\Cref{fig:W}). 

\paragraph{Value calculation.} The final step is to consider the value of this blueprint. In our construction, we are going to ensure that the sizes of each of $\bLstar_{j,t}$ and $\bRstar_{2m-j,t}$ that are matched together are 
(almost) equal to each other. This will ensure that the size of the matching of the blueprint is (almost) $\card{\bLstar}$, which implies that
\begin{align}
	value(\bG) := \frac{\card{\bLstar}}{C^P} = \Pr_{x \in_R [C]^P}\paren{\text{walk $W(x)$ reaches $0$}}.  \label{eq:over-value}
\end{align}
As stated earlier, the constraint we have is that $\bLstar_{j,t}$ and $\bRstar_{2m-j,t}$ sets should have roughly the same size. For now let us \emph{relax} this constraint to instead 
ask that $\bLstar_{j} := \cup_t~ \bLstar_{j,t}$ and $\bRstar_{2m-j} := \cup_t~ \bRstar_{2m-j,t}$ have the same size. This translates to having, 
\begin{align}
	&\forall j \in [2m-1] \Pr_{x \in_R [C]^P}\paren{\text{walk $W(x)$ reaches $0$ and maximum state it visits is $j$}} \notag \\
	&\hspace{90pt}= \Pr_{x \in_R [C]^P}\paren{\text{walk $\barW(x)$ reaches $0$ and maximum state it visits is $2m-j$}}. \label{eq:over-constraint}
\end{align}
The underlying distributions of both $W(x)$ and $\barW(x)$ here are just standard random walks starting from a random state chosen according to 
some distribution $\mu$. Thus, at this point, the goal is to maximize~\Cref{eq:over-value} subject to the constraints in~\Cref{eq:over-constraint}, by choosing
the initial distribution $\mu$ carefully. 

The type of random walk we have here is often called ``gambler's ruin'' and is well studied as one of the simplest examples of Markov Chains; see, e.g.~\cite[Section 2.1]{levin2026markov}. 
For instance, it is known that starting from a \emph{uniform} distribution $\mu$, maximum visited state will also be distributed uniformly and thus certainly satisfies~\Cref{eq:over-constraint}. But in this case, 
the walk has an equal chance of reaching $0$ and $2m$, making $value(\bG)=1/2$ only. However, by ``shifting'' the distribution of $\mu$ towards $0$, while still enforcing~\Cref{eq:over-constraint}, one
can improve the objective function. We are not aware of prior work that study this particular formulation but solving this optimization problem has a standard method: one can cast the problem as a 
linear program in $\mu(i)$-values and obtain a closed form formula for the variables. Doing this in our case leads to the objective value of $2m/(3m+1)$ (see~\Cref{lem:start-state}), giving us 
the desired value for the blueprint. 

There is however one catch with the argument above: we moved from enforcing (almost) equality in sizes of $\bLstar_{j,t}$ and $\bRstar_{2m-j,t}$ to only $\bLstar_j$ and $\bRstar_{2m-j}$; but, symmetrizing the distribution 
of maxima (through the choice of $\mu$) does not mean the time to reaching them will also be symmetric. This however can be handled rather easily by introducing \textbf{random delays} in the random walks: we dedicate one more coordinate in $[P]$
for this task, and start the random walks at different time stages to ensure that distribution of ``time to hit maximum'' becomes (almost) uniform over the time steps which fixes the above problem. 

\subsection{Further Corollaries of Our Result}\label{sec:corollary} 

We mention two further corollaries of our results. The first one is for the matching problem but in another model, the \emph{online matching with preemption}, whereas the second one is for another problem, the \emph{minimum vertex cover}, but back
in the semi-streaming model. 

\paragraph{Online matching with preemption.} In this problem, the edges of the graph are arriving sequentially
and the algorithm must maintain a matching by accepting or rejecting arriving edges and can preempt previously accepted ones by discarding them. This model is a generalization of the original online matching (with one-sided vertex arrival) in~\cite{KarpVV90}. 

The \emph{competitive ratio} of an algorithm in this model is the ratio of the matching size provided by the algorithm compared to the maximum matching of the graph. 
The greedy algorithm achieves a half competitive ratio in this model with no need for preemption. When preemptions are not allowed,~\cite{GamlathKMSW19} proved that  
half approximation is optimal. No better algorithm is known even when preemption is allowed and a series of results over the years have ruled out competitive ratios of 
$(1-1/e) \sim 0.632$~\cite{KarpVV90}, $(1+\ln{(2)})^{-1} \sim 0.590$~\cite{EpsteinLSW13}, $(2-\sqrt{2}) \sim 0.585$~\cite{HuangPTTWZ19}, and very recently $0.5661$~\cite{KissS26}\footnote{It is not a coincidence 
that these numbers and their progressions sound very similar to the ones mentioned for the semi-streaming matching: until~\cite{AssadiJX26}, all previous semi-streaming matching lower bounds in~\cite{GoelKK12,Kapralov13,Kapralov21} 
were based on generalizing online preemptive matching instances to the semi-streaming model (but using completely white-box arguments); see~\cite{Kapralov21} for a detailed discussion of this topic.}.
We prove that the greedy algorithm is also optimal for this problem. 

\begin{corollary}\label{cor:online}
	There is no algorithm for the online matching problem with preemption with an expected competitive ratio 
	strictly better than half by any constant. 
\end{corollary}

As already observed in~\cite{AssadiJX26}, the blueprint framework also implies lower bounds for the online matching problem with preemption (given this connection, technically, the best known bound on the competitive ratios 
prior to our work is the $0.558$ bound of~\cite{AssadiJX26}, which was obtained independently and concurrently to~\cite{KissS26}). Thus,~\Cref{cor:online} directly follows from our~\Cref{thm:blueprint} and the 
blueprint framework of~\cite{AssadiJX26}.

\paragraph{Semi-streaming (bipartite) vertex cover.} Going back to the semi-streaming model, a closely related problem to maximum bipartite matching is its dual, the minimum bipartite vertex cover problem. 
The greedy algorithm is also a $2$-approximation algorithm for this problem and no better algorithms are known. It has been asked in prior work, e.g.,~\cite{AssadiBBMS19,DerakhshanGR25} whether this algorithm can be improved. 
We resolve this question in this work. 

\begin{corollary}\label{cor:vc}
	There is no single-pass semi-streaming algorithm for the minimum bipartite vertex cover problem that can achieve any constant approximation ratio strictly better than two with constant probability. 
\end{corollary}
We note that prior to our work, no $2$-approximation lower bounds were known for this problem even in general graphs (the UGC hardness of $2$-approximation vertex cover~\cite{KhotR03} does not imply a semi-streaming lower bound, 
as runtime of algorithms is not bounded in this model). 

The proof of this result is again by observing that the blueprint framework of~\cite{AssadiJX26} also implies lower bounds for the vertex cover problem given the duality of the problems (plus a standard modification for 
going from matching lower bounds to vertex cover ones as done previously in the literature, e.g., in~\cite{AssadiK17,DerakhshanGR25}).

\clearpage

% !TeX root = main.tex 
%!TEX root = main.tex

\newcommand{\maxwalk}[1]{\ensuremath{\textnormal{\textsf{w}}_{\textnormal{max}}(#1)}}
\newcommand{\timemax}[1]{\ensuremath{\textnormal{\textsf{t}}_{\textnormal{max}}(#1)}}

\newcommand{\maxwalkW}{\ensuremath{\textnormal{\textsf{w}}_{\textnormal{max}}}}
\newcommand{\timemaxW}{\ensuremath{\textnormal{\textsf{t}}_{\textnormal{max}}}}

\newcommand{\dist}{\ensuremath{\mathcal{D}}}

\section{Preliminaries}\label{sec:prelim}

\paragraph{Notation.} For integers $1 \leq a \leq b$, we define $[a:b] := \set{a,a+1,\ldots,b}$ and $[b] := \set{1,\ldots,b}$. For a tuple $x=(x_1,\ldots,x_t)$ and $i \in [2:t]$, 
we denote $x_{<i} := (x_1,\ldots,x_{i-1})$.  

For a random variable $X$, we sometimes write the subscript $X$ in $\Pr_X(f(X))$ (for some $f$), to 
explicitly specify the randomness is coming from $X$.  
For a random variable $X$ and a set $S$, we write $X \in_R S$ to mean $X$ is sampled uniformly at random from $S$. 

\paragraph{Gambler's ruin.} We use a standard type of random walks on integer line $\set{0,1,\ldots,2m}$ for some integer $m \geq 1$, often called the \emph{gambler's ruin}; see, e.g.~\cite[Section 2.1]{levin2026markov}. 

The walk $W=w_1,w_2,\cdots$ starts at some integer $w_1 \in [2m-1]$ and in each time $t \geq 1$, it updates $w_{t+1} = w_{t}+1$ with probability half and $w_{t+1} = w_{t}-1$ otherwise. 
The walk terminates at the first time $t$ when $w_t = 0$ or $w_t = 2m$ and remains in this state. Throughout, by a random walk, we always mean the above type of walk unless explicitly stated otherwise. 

\begin{fact}[{c.f.~\cite[Proposition 2.1]{levin2026markov}}]\label{fact:ruin}
	For a random walk $W$ over $\set{0, \ldots, 2m}$ and an integer $k$ in $\set{0,\ldots,2m}$,
	\[
		\Pr_{W}\paren{\text{$W$ reaches $0$ before $2m$} \mid \text{$W$ starts at $k$}} = \frac{2m-k}{2m}. 
	\]
	Moreover, the expected length of the walk before terminating is at most $m^2$. 
\end{fact}

For a (finite) random walk $W = w_1,\cdots,w_T$, where $w_T \in \set{0,2m}$ is a terminating state, define: 
\begin{align*}
	\maxwalkW &= \maxwalk{W} := \max_{t \in [T]} \set{w_t} \quad \text{and}\quad 
	\timemaxW = \timemax{W} := \min \set{t \in [T] : w_t=\maxwalk{W}},  
\end{align*}
as the \emph{largest} state seen on the walk and the \emph{first} time the walk sees its eventual largest state. 

A corollary of~\Cref{fact:ruin} is the following. 
\begin{claim}\label{clm:ruin}
	For a random walk $W$ over $\set{0,\ldots,2m}$ and any $1 \leq i \leq j < 2m$,
	\[
		\Pr_W\paren{\text{$W$ ends in $0$} \wedge \maxwalk{W} = j \mid \text{$W$ starts at $i$}} = \frac{i}{j \cdot (j+1)}. 
	\]
\end{claim}
\begin{proof}
	For this to happen, we need $W$ to reach $j$ before $0$, and once on $j$, to reach $0$ before $j+1$. Given the random walk is memoryless, these two events are independent and thus we have
	\begin{align*}
		&\Pr_W\paren{\text{$W$ ends in $0$} \wedge \maxwalk{W} = j \mid \text{$W$ starts at $i$}} = \\
		&\hspace{15pt} \Pr_W\paren{\text{$W$ reaches $j$ before $0$} \mid \text{$W$ starts at $i$}} \cdot  \Pr_W\paren{\text{$W$ reaches $0$ before $j+1$} \mid \text{$W$ starts at $j$}}. 
	\end{align*}
	The first term is about a random walk on $\set{0,\ldots,j}$ starting from $i$ and thus by~\Cref{fact:ruin} has probability $i/j$. The second term is about a random walk on $\set{0,\ldots,j+1}$ starting from $j$
	and thus by~\Cref{fact:ruin} has probability $1/(j+1)$. Multiplying these two concludes the proof. 
\end{proof}
We also have the following standard claim on the length of random walks. 
\begin{claim}\label{clm:walk-length}
	The probability that a random walk $W$ starting from any state does not terminate within $2m^3$ steps is at most $2^{-m}$. 
\end{claim}
\begin{proof}
	The expected time for a walk $W$ to terminate, starting from any state, is at most $m^2$ by~\Cref{fact:ruin}. 
	Thus, by Markov bound, the probability that the walk has not terminated after $2m^2$ steps is at most half. Repeating the argument from current state, 
	the probability that $m$ repetitions, each of length $2m^2$, does not terminate the walk is at most $2^{-m}$. 
\end{proof}

\clearpage

% !TeX root = main.tex
%!TEX root = main.tex

\section{Random Walk on a Line with Symmetric Maxima}\label{sec:random-walk}

Fix any integer $m \geq 1$. We will be using the gambler's ruin random walk on $\set{0,\ldots,2m}$ described in~\Cref{sec:prelim} for our proofs. 
For our purpose, we also need to specify a distribution over initial states with the following two properties: (1) the distribution of $\maxwalk{W}$ is symmetric (on the interior states) assuming the start state is sampled from this distribution, 
and, (2) the probability of terminating in state $0$ is maximized. The following lemma finds such a distribution for us (we do not explicitly prove that it maximizes the term in $(2)$ and only 
calculate its value, which is sufficient for us; however, one can prove it is maximized). 

\begin{lemma}\label{lem:start-state}
	Consider the assignment $\mu: \set{1,\ldots,2m-1} \rightarrow [0,1]$ such that for $i \in [m]$,
	\[
		\mu(i) := \gamma_{m} \cdot \frac{2 \cdot (2m+1)}{(2m-i)\cdot(2m-i+1)\cdot(2m-i+2)} \quad \text{for} \quad \gamma_m := \frac{2m \cdot (m+1)}{3m+1},
	\]
	and for $i \in [m+1:2m-1]$, $\mu(i)=0$. 
	
	Firstly, $\mu$ is a distribution. Secondly, if we sample a random walk $W$ starting from a state distributed according to $\mu$---denoting the combined distribution by $\dist_m$---, then, for all $j \in [2m-1]$, 
	\[
		\Pr_{W \sim \dist_m}\paren{\text{$W$\!~ends in $0$} ~\wedge~ \maxwalk{W} = j} = \Pr_{W \sim \dist_m}\paren{\text{$W$\!~ends in $0$} ~\wedge~ \maxwalk{W} = 2m-j}. 
	\]
	Finally, 
	\[
		\Pr_{W \sim \dist_m}\paren{\text{$W$\!~ends in $0$}} = \frac{2m}{3m+1}. 
	\]
\end{lemma}

We postpone the proof of this lemma to the end of this section to complete our definitions first. The proof can be skipped by the reader as it is not needed for understanding the rest of the paper. 

\paragraph{Delayed and length-restricted random walks.} \Cref{lem:start-state} ``symmetrizes'' the distribution of maxima for a random walk $W$. We also need to ensure the distribution of the time of reaching 
the maxima is (essentially) the same across different choices for $\maxwalk{W}$. Adding a ``delay'' to the walk fixes this issue as well.  

A \textbf{delayed random walk} with delay $d \geq 1$ (also called $d$-delayed random walk) is defined similar to our original random walk, except that for its first $d-1$ time stamps, the walk 
stays in the same state and starts moving as before from time step $d$. Specifically, a $d$-delayed random walk $W=w_1,w_2,\cdots,$ is such that 
\[
	w_1 = \cdots = w_d \quad \text{and} \quad \text{for $t \geq d$,}~ w_{t+1}~\text{is chosen uniformly from}~\set{w_t + 1 , w_t-1}, 
\]
and the walk as before terminates when hitting the state $0$ or $2m$. We refer to the ``non-delayed'' part of the walk $W$ as the \emph{inner} random walk $W_{in}$, namely, $W_{in} = w_d, w_{d+1},\cdots$. 

For a $d$-delayed random walk $W$, we define 
\begin{align}
\timemax{W} = (d-1)+\timemax{W_{in}}, \label{eq:timemaxW}
\end{align}
 namely, the time the \emph{inner} walk reached the eventual maximum (this distinction between inner and the entire walk is only relevant when the starting state is the eventual maximum; in this case, 
we consider $\timemax{W}=d$, namely, the delay of the walk instead of time $1$); $\maxwalk{W}$ is defined as before which is equal to $\maxwalk{W_{in}}$ regardless. 

Finally, we also need the length of the walks to be finite, which can be guaranteed (essentially) by~\Cref{clm:walk-length}.
We say that a delayed random walk is \textbf{$\ell$-length restricted} iff 
it terminates within $\ell$ steps \emph{after} spending its delayed time, i.e., reaches states $0$ or $2m$ within $d+\ell$ steps assuming its delay was $d$ (or alternatively, the length of its inner walk $W_{in}$ is at most $\ell$). 
By~\Cref{clm:walk-length}, for large enough $\ell$, a (delayed) random walk will be $\ell$-length restricted with high probability (in $m$).

\subsection{Proof of~\Cref{lem:start-state}}
We remark that the ``right'' way to prove a statement of the type~\Cref{lem:start-state} is to
solve the program that maximizes $\Pr(\text{$W$ ends in $0$})$ subject to the symmetric constraints, under variables $\set{\mu(i)}_{i=1}^{2m-1}$ which are non-negative and sum to one. 
This becomes a linear programming problem which can then be solved explicitly (by considering its dual), to determine the values of $\mu(i)$'s. Nevertheless, since we have
already solved this program and obtained $\mu(i)$'s, in the following proof we only verify the correctness of these variables. 

We prove each part of~\Cref{lem:start-state} in the following. 

\paragraph{$\mu$ forms a distribution.} Clearly, all $\mu(i) \geq 0$ so we only need to prove $\sum_i \mu(i)=1$. For this, 
\begin{align*}
	\sum_i \mu(i) &= \gamma_m \cdot \sum_{i=1}^{m}  \frac{2 \cdot (2m+1)}{(2m-i)\cdot(2m-i+1)\cdot(2m-i+2)} \\
	&= \gamma_m \cdot \sum_{r=m}^{2m-1} \frac{2 \cdot (2m+1)}{r \cdot (r+1)\cdot(r+2)} \tag{by reindexing the sum with $r=2m-i$} \\
	&= \gamma_m \cdot \sum_{r=m}^{2m-1} (2m+1) \cdot \paren{\frac{1}{r \cdot (r+1)} - \frac{1}{(r+1)\cdot(r+2)}} \tag{by a direct calculation} \\
	&= \gamma_m \cdot (2m+1) \cdot \paren{\frac{1}{m \cdot (m+1)} - \frac{1}{(2m)\cdot(2m+1)}} \tag{by the telescoping sum} \\
	&= \gamma_m \cdot \frac{3m+1}{2m \cdot (m+1)} = 1, 
\end{align*}
by the definition of $\gamma_m$. Thus, $\mu$ is a distribution. 

\paragraph{Symmetric property.} To simplify the notation, in the following, we use `$w_T=0$' to denote the event that the walk ends in $0$ (instead of $2m$) and use `$w_1=i$' for any integer $i$ to mean
the walk starts from the state $i$. For any $j \in [2m-1]$, 
\begin{align}
	\Pr_{W \sim \dist_m}\paren{w_T=0 ~\wedge~ \maxwalk{W} = j} &= \sum_{i=1}^{m} \mu(i) \cdot \Pr\paren{w_T=0 ~\wedge~ \maxwalk{W} = j \mid w_1=i} \tag{as the 
	walk needs to start somewhere in $[m]$} \\
	&=  \sum_{i=1}^{\min(m,j)} \mu(i) \cdot \Pr\paren{w_T=0 ~\wedge~ \maxwalk{W} = j \mid w_1=i} \tag{as otherwise starting from $i > j$ ensures $\maxwalk{W}=j$ never happens} \\
	&=  \sum_{i=1}^{\min(m,j)} \mu(i) \cdot \frac{i}{j \cdot (j+1)}, \label{eq:get-back-to}
\end{align}
by~\Cref{clm:ruin} since $1 \leq i \leq j \leq 2m$. 
We have the following claim on partial sums  above. 
\begin{claim}\label{clm:partial-sum}
	For any $ j \in [m]$, 
	\[
		P_j := \sum_{i=1}^{j} \mu(i) \cdot i = \gamma_m \cdot \frac{j \cdot (j+1)}{(2m-j) \cdot (2m-j+1)}. 
	\]
\end{claim}
\begin{proof}
	The proof is by induction. For $j=1$, 
	\[
		P_1 = \mu(1) = \gamma_{m} \cdot \frac{2 \cdot (2m+1)}{(2m-1)\cdot(2m)\cdot(2m+1)} = \gamma_{m} \cdot \frac{1 \cdot (1+1)}{(2m-1)\cdot(2m)}, 
	\]
	matching the RHS of the claim. For $1  < j \leq m$, for the induction step, by the value of $\mu(j)$, 
	\begin{align*}
		 P_j = P_{j-1} + j \cdot \mu(j) &= \gamma_m \cdot \frac{(j-1) \cdot j}{(2m-j+1) \cdot (2m-j+2)} + j \cdot  \gamma_{m} \cdot \frac{2 \cdot (2m+1)}{(2m-j)\cdot(2m-j+1)\cdot(2m-j+2)} \\
		&= \gamma_m \cdot \frac{{j \cdot (j-1) \cdot (2m-j)+2j \cdot (2m+1)}}{(2m-j)\cdot(2m-j+1)\cdot(2m-j+2)} \tag{by multiplying the first fraction with $(2m-j)/(2m-j)$}\\
		&= \gamma_m \cdot \frac{j \cdot (j+1)}{(2m-j) \cdot (2m-j+1)}, 
	\end{align*}
	where the last equation can be verified using a direct calculation by verifying that 
	\[
		(2m-j+2) \cdot (j+1) = (j-1) \cdot (2m-j)+2\cdot (2m+1). 
	\] 
	This concludes the proof. \Qed{clm:partial-sum}
	
\end{proof}
Plugging in~\Cref{clm:partial-sum} in~\Cref{eq:get-back-to}, we have that for any $j \in [m]$, 
\[
	\Pr_{W \sim \dist_m}\paren{w_T=0 ~\wedge~ \maxwalk{W} = j} = \frac{P_j}{j \cdot (j+1)} = \gamma_m \cdot \frac{1}{(2m-j) \cdot (2m-j+1)}.
\]
On the other hand, for $j \in [m]$, we have $2m-j \geq m$ and by the max-term in~\Cref{eq:get-back-to}, 
\[
	\Pr_{W \sim \dist_m}\paren{w_T=0 ~\wedge~ \maxwalk{W} = 2m-j} = \frac{P_m}{(2m-j) \cdot (2m-j+1)} = \gamma_m \cdot \frac{1}{(2m-j) \cdot (2m-j+1)}, 
\]
since by~\Cref{clm:partial-sum}, $P_m =\gamma_m$. This proves the equality in the symmetric property. 

\paragraph{Probability of $W$ ending in $0$.} We calculate this directly as 
\begin{align*}
	\Pr_{W \sim \dist_m}\paren{w_T=0} &= \sum_{i=1}^{m} \mu(i) \cdot \Pr_{W}\paren{w_T=0 \mid w_1=i}	\tag{by the distribution of the start state} \\
	&= \sum_{i=1}^{m} \mu(i) \cdot \Pr_{W}\paren{\text{$W$ reaches $0$ before $2m$} \mid w_1=i} \tag{as otherwise the walk terminates in $2m$ instead} \\
	&= \sum_{i=1}^{m} \mu(i) \cdot \frac{2m-i}{2m} \tag{by~\Cref{fact:ruin}} \\
	&= 1- \sum_{i=1}^{m} \frac{i \cdot \mu(i)}{2m} \tag{since $\mu$ is a distribution so $\sum_i \mu(i)=1$} \\
	&= 1- \frac{\gamma_m}{2m} \tag{as the numerator in the sum-term is $P_m$ which is $\gamma_m$ by~\Cref{clm:partial-sum}} \\
	&= \frac{2m}{3m+1},
\end{align*}
by the value of $\gamma_m$. This concludes the proof of~\Cref{lem:start-state}. \qed

\clearpage

% !TeX root = main.tex
%!TEX root = main.tex

\section{The Blueprint Construction}\label{sec:construct-blueprint}

We provide our blueprint construction in this section, proving the following theorem. 
\begin{theorem*}[Restatement of~\Cref{thm:blueprint}]
	For any sufficiently large integer $m$, there exists a proper blueprint $\bG_m$ with parameters $P=\poly(m)$ and $C=m^{O(m)}$ such that 
	\[
		value(\bG_m) = \frac{2}{3} - \Theta\!\paren{\frac{1}{m}}. 
	\]
\end{theorem*}

We start by presenting the setup needed for this construction, followed by the construction itself and its properness proof, and finally calculate its value to conclude the proof of~\Cref{thm:blueprint}. 

\subsection{Setup} 
Let $m$ be a large integer that is used to define all our other parameters for the blueprint. We define a blueprint $\bG := \bG(m)$ in this section. 
Define: 
\begin{align*}
L &:= m^4, \tag{used to determine the \emph{length} of random walks} \\
D &:= 2m \cdot L, \tag{used to determine the (range of) \emph{delays} in random walks} \\
P &:= L+D+1 \tag{parameter $P$ of the blueprint} \\
C &:= \text{sufficiently large as function of $m$ to be fixed later} \tag{parameter $C$ of the blueprint}. 
\end{align*}
Recall that vertices of $\bG$ are $\bL,\bR$ and they are labeled by $[C]^{P}$. We partition indices of $P$ as follows: 
\begin{gather*}
	P := (\pstart , \pdelay , \pwalk) \quad \text{where} \\ 
	\pstart=1, \quad \pdelay=2, \quad \text{and,} \quad {\pwalk} = [3: P]. 
\end{gather*}

We use these indices to associate a delayed random walk to vertices of the blueprint as follows. 

\begin{definition}\label{def:W(x)}
Sample $x \in [C]^P$ uniformly at random and define a \textbf{delayed random walk} $W = W(x) = w_1,w_2,\cdots,w_T$ for some $T \leq D+L$ as follows\footnote{It is possible that the walk does not terminate and in that case we only consider its $D+L$ first states.}.
\begin{enumerate}[itemsep=3pt]
\item\label{line:pstart} $x_{\pstart}$ determines the start state $w_1$ of the random walk $W$: we partition $[C]$ into intervals of proportional length $\mu(1),\mu(2),\ldots,\mu(m)$ (as defined in~\Cref{lem:start-state}), 
and pick $w_1$ depending on which interval $x_{\pstart}$ lands in. 
\item\label{line:pdelay} $x_{\pdelay}$ determines the delay of the delayed random walk $W$: we partition $[C]$ into $D$ equal-length intervals and the delay of the random walk $d(x)=d$
if $x_{\pdelay}$ lands in $d$-th interval.  
\item\label{line:reflection} $x_{\pwalk}$ determines the inner walk: for any time $t \geq d(x)$, $x_{t+2}$ determines the transition from $w_t$ to $w_{t+1}$ in $W$, where
$w_{t+1} = w_t -1$ if $x_{t+2} \in [1:{C}/{2}]$ and $w_{t+1}=w_t + 1$ otherwise.
The walk terminates whenever it hits $0$ or $2m$ and the subsequent states remain the same. 
\end{enumerate}

We also define the \textbf{reflection walk} of $x$, denoted by $\barW(x)$, exactly as $W(x)$ except for Line~\eqref{line:reflection}: we now increase the state if $x_i \in [1:C/2]$ and decrease it otherwise (the exact opposite rule as in the definition of $W$).  
\end{definition}
See~\Cref{fig:W(x)} for an illustration of this definition and further help with parsing it. 

\clearpage
Before getting to the main part, a quick remark is in order. 
\begin{remark}
We pick $C$ such that 
\begin{itemize}
	\item $\mu(i) \cdot C$ is an integer for all $i \in [2m-1]$ (so the distribution of the start point in Line~\eqref{line:pstart} 
	matches $\mu$);  
	\item $D | C$ (so the distribution of the delay in Line~\eqref{line:pdelay} is uniform over $[D]$); 
	\item $C$ is even (so the distribution of the inner walk in Line~\eqref{line:reflection} is a standard random walk).  
\end{itemize}
	We can thus set 
	\[
		C = 2D \cdot \prod_{i=1}^{m} \textnormal{denominator of $\mu(i)$}
	\]
	and given the values of parameters $\mu(i)$ in~\Cref{lem:start-state} and since $D= \poly(m)$, we have $C=m^{O(m)}$. 
\end{remark}

\begin{figure}[t!]
\centering
\begin{tikzpicture}[
    x=0.78cm,
    y=1cm,
    >=Latex,
    line cap=butt
]

    %------------------------------------------------------------
    % Geometry
    %------------------------------------------------------------
    \def\boxheight{2.8}
    \def\boxwidth{16}

    % End of the first and second boxes
    \def\startend{1.6}
    \def\delayend{3.6}

    % Partitions inside the walk region
    \def\walkA{6.4}
    \def\walkB{7.55}
    \def\walkC{8.70}
    \def\walkD{9.85}

    % Half-length of each small horizontal interval separator
    \def\separatorhalfwidth{0.12}

    % Vertical midpoint of the large box
    \pgfmathsetmacro{\ymid}{\boxheight/2}

    % Centers of the first two boxes
    \pgfmathsetmacro{\startcenter}{\startend/2}
    \pgfmathsetmacro{\delaycenter}{(\startend+\delayend)/2}

    % Centers of the three narrow boxes
    \pgfmathsetmacro{\dcenterA}{(\walkA+\walkB)/2}
    \pgfmathsetmacro{\dcenterB}{(\walkB+\walkC)/2}
    \pgfmathsetmacro{\dcenterC}{(\walkC+\walkD)/2}

    % The walk region is everything after the delay box
    \pgfmathsetmacro{\walkcenter}{(\delayend+\boxwidth)/2}

    %------------------------------------------------------------
    % Styles
    %------------------------------------------------------------
    \tikzset{
        outer box/.style={
            draw=black,
            rounded corners=2mm,
            line width=0.65pt
        },
        partition/.style={
            draw=black,
            line width=0.55pt
        },
        colored bar/.style={
            line width=3pt,
            line cap=butt
        },
        interval separator/.style={
            draw=black,
            line width=0.55pt,
            line cap=round
        },
        small arrow/.style={
            -{Latex[length=1.8mm,width=1.35mm]},
            line width=0.9pt,
            black
        },
        upward arrow/.style={
            -{Latex[length=2.7mm,width=2.1mm]},
            line width=1.5pt,
            green!65!black,
            line cap=butt
        },
        downward arrow/.style={
            -{Latex[length=2.7mm,width=2.1mm]},
            line width=1.5pt,
            red!70,
            line cap=butt
        },
        lower label/.style={
            anchor=north,
            font=\large\itshape
        },
        set label/.style={
            font=\Large
        }
    }

    %------------------------------------------------------------
    % Clip all interior elements to the rounded outer box
    %------------------------------------------------------------
    \begin{scope}
        \clip[rounded corners=2mm]
            (0,0) rectangle (\boxwidth,\boxheight);

        %--------------------------------------------------------
        % Exact top-to-bottom partition lines
        %--------------------------------------------------------
        \foreach \x in {
            \startend,
            \delayend,
            \walkA,
            \walkB,
            \walkC,
            \walkD
        }{
            \draw[partition]
                (\x,0)
                --
                (\x,\boxheight);
        }

        %--------------------------------------------------------
        % First box: unequal interval lengths
        %--------------------------------------------------------
        \draw[colored bar,green!65!black]
            (\startcenter,0)
            --
            (\startcenter,0.55);

        \draw[colored bar,red!75]
            (\startcenter,0.55)
            --
            (\startcenter,1.90);

        \draw[colored bar,blue!60]
            (\startcenter,1.90)
            --
            (\startcenter,\boxheight);

        % Horizontal separators between the three intervals
        \foreach \y in {0.55,1.90}{
            \draw[interval separator]
                ({\startcenter-\separatorhalfwidth},\y)
                --
                ({\startcenter+\separatorhalfwidth},\y);
        }

        %--------------------------------------------------------
        % Second box: six equal intervals
        %--------------------------------------------------------
        \foreach \col [count=\i from 0] in {
            red!75,
            orange!85!black,
            cyan!65!black,
            green!60!black,
            magenta!65,
            blue!60
        }{
            \pgfmathsetmacro{\segmentbottom}{\i*\boxheight/6}
            \pgfmathsetmacro{\segmenttop}{(\i+1)*\boxheight/6}

            \draw[colored bar,\col]
                (\delaycenter,\segmentbottom)
                --
                (\delaycenter,\segmenttop);
        }

        % Horizontal separators between the six equal intervals
        \foreach \i in {1,...,5}{
            \pgfmathsetmacro{\separatorheight}
                {\i*\boxheight/6}

            \draw[interval separator]
                ({\delaycenter-\separatorhalfwidth},
                 \separatorheight)
                --
                ({\delaycenter+\separatorhalfwidth},
                 \separatorheight);
        }

        %--------------------------------------------------------
        % Three narrow boxes:
        % each is divided into equal upper and lower intervals
        %--------------------------------------------------------
        \foreach \x in {
            \dcenterA,
            \dcenterB,
            \dcenterC
        }{
            % Upper interval
            \draw[upward arrow]
                (\x,\ymid)
                --
                (\x,\boxheight);

            % Lower interval
            \draw[downward arrow]
                (\x,\ymid)
                --
                (\x,0);

            % Horizontal separator between the two intervals
            \draw[interval separator]
                ({\x-\separatorhalfwidth},\ymid)
                --
                ({\x+\separatorhalfwidth},\ymid);
        }

        %--------------------------------------------------------
        % Small left-to-right arrows
        %--------------------------------------------------------

        % First box
        \draw[small arrow]
            (0.32,1.23)
            --
            (0.59,1.23);

        % Second box
        \draw[small arrow]
            (2.08,2.15)
            --
            (2.35,2.15);

        % First narrow box
        \draw[small arrow]
            (6.57,2.02)
            --
            (6.82,2.02);

        % Second narrow box
        \draw[small arrow]
            (7.70,0.70)
            --
            (7.95,0.70);

        % Third narrow box
        \draw[small arrow]
            (8.85,0.95)
            --
            (9.10,0.95);
    \end{scope}

    %------------------------------------------------------------
    % Draw the outer boundary last
    %------------------------------------------------------------
    \draw[outer box]
        (0,0)
        rectangle
        (\boxwidth,\boxheight);

    %------------------------------------------------------------
    % External labels
    %------------------------------------------------------------
\node[set label,overlay]
    at (-0.72,\ymid)
    {$[C]$};

    \node[set label]
        at ({\boxwidth/2},{\boxheight+0.62})
        {$[P]$};

    % All lower labels share the same baseline
    \node[lower label]
        at (\startcenter,-0.1)
        {\pstart};

    \node[lower label]
        at (\delaycenter,-0.1)
        {\pdelay};

    \node[lower label]
        at (\dcenterA,-0.1)
        {$d+2$};

    \node[lower label]
        at (\walkcenter,-0.1)
        {\pwalk};

\end{tikzpicture}

\caption{An illustration of the random walk $W(x)$ in~\Cref{def:W(x)} when $m=3$ and $D=6$. The figure shows how the indices of $[P]$ are partitioned into $(\pstart,\pdelay,\pwalk)$. The intervals inside each box shows how $[C]$ is partitioned into different intervals. The tiny arrows in each box shows the value of a single point $x \in [C]^P$ on each relevant index. The walk $W(x)$ thus first considers $x_{\pstart}$ and so chooses starting state $2$. The delay $d$ is set to $5$ given $x_{\pdelay}$. Finally, the walk waits until time $d$ and then decide its next move based on the index $d+2$ and so on; in this particular example, the walk goes to $3$ next, then $2$, and then $1$. }
\label{fig:W(x)}
\end{figure}
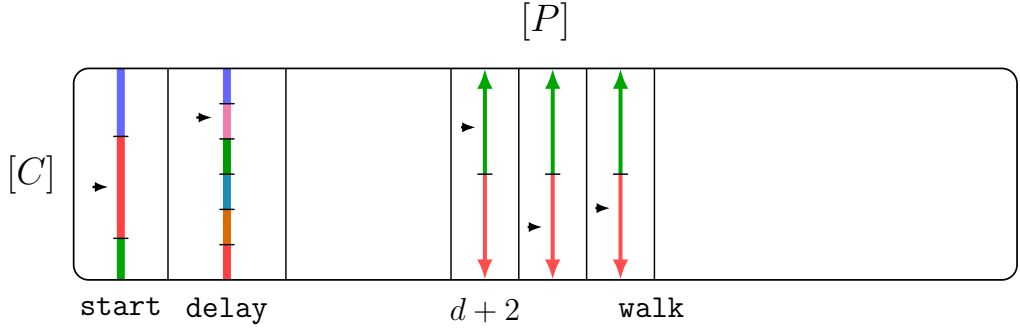

We use walks $W(x)$ for defining subsets of $\bL$ and reflection walks $\barW(x)$ for subsets in $\bR$. Specifically, for any integer $j \in [2m-1]$ and time $t \in [D+L]$, we define
\begin{align*}
	\bL_{j,t} &:= \set{\bL(x) : \text{$W(x)$ is $L$-restricted, $W(x)$ reaches $0$, $\maxwalk{W(x)}=j$, and $\timemax{W(x)} = t$}}, \\
	\bR_{j,t} &:= \set{\bR(x) : \text{$\barW(x)$ is $L$-restricted, $\barW(x)$ reaches $0$, $\maxwalk{\barW(x)}=j$, and $\timemax{\barW(x)} = t$}}. 
\end{align*}
These sets will play the key role in the definition of edges of the blueprint (they are essentially the endpoints of the edges). 

We have the following immediate relation between the sizes of these sets. 

\begin{claim}\label{obs:bjt-direct}
	For all $j \in [2m-1]$ and $t \in [D+L]$, $\card{\bL_{j,t}} = \card{\bR_{j,t}}$.
\end{claim}
\begin{proof}
	The only difference between $\bL_{j,t}$ and $\bR_{j,t}$ is in Line~\eqref{line:reflection} of~\Cref{def:W(x)}: we thus have $\bL(x) \in \bL_{j,t}$ iff $\bR(\overline{x}) \in \bR_{j,t}$, where
	$\overline{x}$ is defined to be the same as $x$ in $[P] \setminus \pwalk$ and for $i \in \pwalk$, $\overline{x}_i := C - x_i + 1$. Using the bijection between $x$'s and $\overline{x}$'s, we can conclude the proof. 
\end{proof}

The more interesting property is the near-equality of the sizes of $\bL_{j,t}$ and $\bR_{2m-j,t}$ (whenever $t \in [L:D]$); this is where we use most of the properties we incorporated in our random walks. 
\begin{lemma}\label{lem:bjt-reflect}
	For all $j \in [2m-1]$ and $t \in [L:D]$, 
	\[
	\card{\bR_{2m-j,t}} - 2^{-m} \cdot \frac{C^P}{D} \leq \card{\bL_{j,t}} \leq \card{\bR_{2m-j,t}} + 2^{-m} \cdot \frac{C^P}{D}. 
	\]
\end{lemma}
\begin{proof}
	Define $W_{in}(x)$ as the inner walk of $W(x)$, which is the walk that starts at time $d(x)$ (the delay) and continues the same as $W$; we define $\barW_{in}(x)$ similarly. 
	In the following, when clear from the context, we drop $(x)$ and write $W$ to denote $W(x)$ (similarly for $W_{in},\barW$, etc.). 
	
	We first have 
	\begin{align*}
		\frac{\card{\bL_{j,t}}}{C^P} &= \Pr_{x}\Paren{\text{$W$ $L$-restricted, $W$ reaches $0$, $\maxwalk{W}=j$, $\timemax{W} = t$}} \tag{by~\Cref{def:W(x)}}\\
		&=  \Pr_{x}\Paren{\text{$W_{in}$ $L$-restricted, $W_{in}$ reaches $0$, $\maxwalk{W_{in}}=j$,  $\timemax{W} = t$}} \tag{$L$-restrictedness, reaching $0$, and the max-states are all the same 
		between $W_{in}$ and $W$}  \\
		&= \Pr_{x}\Paren{\text{$W_{in}$ $L$-restricted, $W_{in}$ reaches $0$, $\maxwalk{W_{in}}=j$, $\timemax{W_{in}} = t-(d(x)-1)$}} \tag{by the definition of $\timemax{W}$ in~\Cref{eq:timemaxW}} \\
		&= \frac1D \cdot \Pr_{x}\Paren{\text{$W_{in}$ $L$-restricted, $W_{in}$ reaches $0$, $\maxwalk{W_{in}}=j$}} 
	\end{align*}
	as $d(x)$ is independent of $W_{in}$ and is uniform over $[D]$, 
	and so the event $d(x) = t-\timemax{W_{in}}+1$ (conditioned on $W_{in}$) happens with probability $1/D$ as the RHS is in $[D]$ as well; this is because $\timemax{W_{in}} \leq L$ (by $L$-restrictedness) and $t \in [L:D]$. 
	
	We can do the same exact analysis for $\bR_{2m-j,t}$ and $\barW$ as well to have 
	\[
		\frac{\card{\bR_{2m-j,t}}}{C^P} = \frac{1}{D} \cdot \Pr_{x}\Paren{\text{$\barW_{in}$ $L$-restricted, $\barW_{in}$ reaches $0$, $\maxwalk{\barW_{in}}=2m-j$}}. 
	\]
	Consider the expression in the RHS of either of these two terms. The distribution of $W_{in}$ or $\barW_{in}$ is exactly the same as the distribution of a random walk $\widetilde{W} \sim \dist_m$ defined in~\Cref{lem:start-state}. Hence,  
	\begin{align*}
		\frac{\card{\bL_{j,t}}}{C^P} &=  \frac{1}{D} \cdot \Pr_{\widetilde{W}}\Paren{\text{$\widetilde{W}$ $L$-restricted, $\widetilde{W}$ reaches $0$, $\maxwalk{\widetilde{W}}=j$}} \\
		&\geq \frac{1}{D} \cdot \paren{\Pr_{\widetilde{W}}\Paren{\text{$\widetilde{W}$ reaches $0$, $\maxwalk{\widetilde{W}}=j$}} - \Pr_{\widetilde{W}}\Paren{\text{$\widetilde{W}$ not $L$-restricted}}} \tag{as $\Pr(A \wedge B) \geq \Pr(B) - \Pr(\overline{A})$} \\
		&\geq \frac{1}{D} \cdot \paren{\Pr_{\widetilde{W}}\Paren{\text{$\widetilde{W}$ reaches $0$, $\maxwalk{\widetilde{W}}=j$}} - 2^{-m}} \tag{as $L=m^4$ and by~\Cref{clm:walk-length}} \\ 
		&= \frac{1}{D} \cdot \paren{\Pr_{\widetilde{W}}\Paren{\text{$\widetilde{W}$ reaches $0$, $\maxwalk{\widetilde{W}}=2m-j$}} - 2^{-m}} \tag{by the symmetry property of~\Cref{lem:start-state}} \\
		&\geq \frac{1}{D} \cdot \paren{\Pr_{\widetilde{W}}\Paren{\text{$\widetilde{W}$ $L$-restricted, $\widetilde{W}$ reaches $0$, $\maxwalk{\widetilde{W}}=2m-j$}} - 2^{-m}} \tag{as $\Pr(A) \geq \Pr(A \wedge B)$} \\
		&= \frac{\card{\bR_{2m-j,t}}}{C^P} - \frac{2^{-m}}{D}. 
	\end{align*}
	Applying the same argument starting from $\bR_{2m-j,t}$ then implies that 
	\[
		\card{\bR_{2m-j,t}} - 2^{-m} \cdot \frac{C^P}{D} \leq \card{\bL_{j,t}} \leq \card{\bR_{2m-j,t}} + 2^{-m} \cdot \frac{C^P}{D} \qed
	\]
	
\end{proof}

\subsection{The Construction} 

We are now ready to define the edges of the blueprint $\bG=(\bL,\bR,\bE)$ and prove its properness. As stated earlier, both vertex-sets $\bL$ and $\bR$ are labeled by $[C]^P$. The edges in $\bE$ are defined as follows: 
\begin{itemize}
	\item For any $j \in [2m-1]$ and $t \in [L:D]$, we add a matching $\bM_{j,t}$ of size $s_{j,t} := \min\paren{\card{\bL_{j,t}},\card{\bR_{2m-j,t}}}$ between two arbitrarily-chosen $s_{j,t}$-subsets of $\bL_{j,t}$ and $\bR_{2m-j,t}$. Moreover, 
	the matching $\bM_{j,t}$ is added as part of $\bEE{p}$ for $p=t+2$.
\end{itemize}

We need to check that the blueprint $\bG$ is proper, which is handled by the following two claims. 

\begin{claim}[``Matching constraint'']\label{clm:bG-matching}
	The edge-set $\bE$ is a matching. 
\end{claim}
\begin{proof}
	By definition, the sets $\bL_{j,t}$ are disjoint for different values of $j$ (the maximum of a random walk is a unique state) as well as $t$ (the \emph{first} time reaching the eventual maximum is unique). Same is true for $\bR_{2m-j,t}$ as well, 
	thus $\bE = \cup_{p \in [P]} \bEE{p}$ itself is a matching also.  
\end{proof}

% In the preamble:
% \usepackage{tikz}
% \usetikzlibrary{arrows.meta}

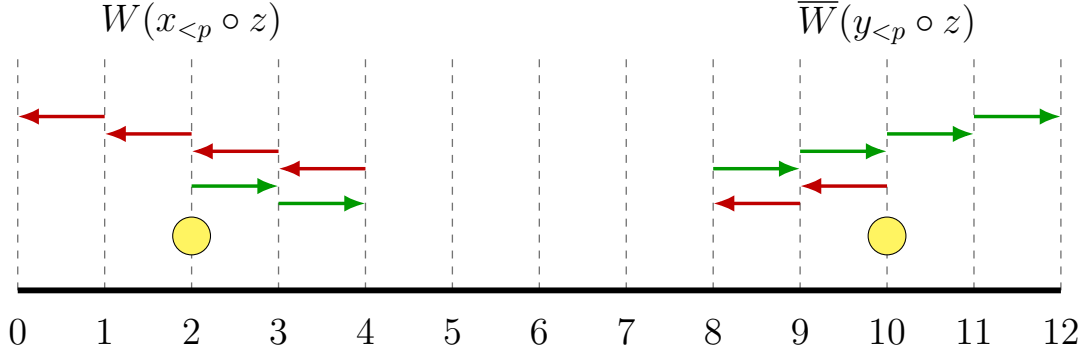
\begin{figure}[t!]
\centering
\begin{tikzpicture}[x=1.15cm,y=1cm,>=Latex]

    %------------------------------------------------------------
    % Parameters
    %------------------------------------------------------------
    \def\axisy{0}
    \def\topy{3.15}
    \def\texty{3.55}
    \def\circley{0.72}
    \def\dy{0.23}

    % Arrow levels, with fixed vertical spacing
    \pgfmathsetmacro{\ya}{2.30}
    \pgfmathsetmacro{\yb}{\ya-\dy}
    \pgfmathsetmacro{\yc}{\yb-\dy}
    \pgfmathsetmacro{\yd}{\yc-\dy}
    \pgfmathsetmacro{\ye}{\yd-\dy}
    \pgfmathsetmacro{\yf}{\ye-\dy}

    %------------------------------------------------------------
    % Styles
    %------------------------------------------------------------
    \tikzset{
        gridline/.style={
            dashed,
            draw=black!55,
            line width=0.5pt
        },
        baseline/.style={
            draw=black,
            line width=2.2pt
        },
        redarrow/.style={
            -{Latex[length=2.8mm,width=2.2mm]},
            draw=red!75!black,
            line width=1.3pt,
            shorten <=0pt,
            shorten >=0pt
        },
        greenarrow/.style={
            -{Latex[length=2.8mm,width=2.2mm]},
            draw=green!60!black,
            line width=1.3pt,
            shorten <=0pt,
            shorten >=0pt
        },
        ball/.style={
            circle,
            draw=black,
            fill=yellow!75,
            minimum size=5mm,
            inner sep=0pt
        }
    }

    %------------------------------------------------------------
    % Dashed vertical lines at every integer
    %------------------------------------------------------------
    \foreach \i in {0,...,12}{
        \draw[gridline]
            (\i,\axisy)
            --
            (\i,\topy);
    }

    %------------------------------------------------------------
    % Horizontal baseline
    %------------------------------------------------------------
    \draw[baseline]
        (0,\axisy)
        --
        (12,\axisy);

    %------------------------------------------------------------
    % Integer labels
    %------------------------------------------------------------
    \foreach \i in {0,...,12}{
        \node[font=\Large]
            at (\i,-0.55)
            {\i};
    }

    %------------------------------------------------------------
    % Circles and upper labels
    %------------------------------------------------------------
    \node[ball] at (2,\circley) {};
    \node[ball] at (10,\circley) {};

    \node[font=\Large]
        at (2,\texty)
        {$W(x_{<p}\circ z)$};

    \node[font=\Large]
        at (10,\texty)
        {$\barW(y_{<p}\circ z)$};

    %------------------------------------------------------------
    % Left side
    % Every arrow has length exactly one grid interval.
    %------------------------------------------------------------

    % Red arrows point left
    \draw[redarrow] (1,\ya) -- (0,\ya);
    \draw[redarrow] (2,\yb) -- (1,\yb);
    \draw[redarrow] (3,\yc) -- (2,\yc);
    \draw[redarrow] (4,\yd) -- (3,\yd);

    % Green arrows point right
    \draw[greenarrow] (2,\ye) -- (3,\ye);
    \draw[greenarrow] (3,\yf) -- (4,\yf);

    %------------------------------------------------------------
    % Right side
    % Reflected through x=6, with colors exchanged.
    %------------------------------------------------------------

    % Green arrows point right
    \draw[greenarrow] (8,\yd) -- (9,\yd);
    \draw[greenarrow] (9,\yc) -- (10,\yc);
    \draw[greenarrow] (10,\yb) -- (11,\yb);
    \draw[greenarrow] (11,\ya) -- (12,\ya);

    % Red arrows point left
    \draw[redarrow] (10,\ye) -- (9,\ye);
    \draw[redarrow] (9,\yf) -- (8,\yf);

\end{tikzpicture}

\caption{An illustration of the proof of the ban constraint (\Cref{clm:bG-ban}) for $\bG$ with $m=6$, $j=2$, and $z \approx (\uparrow,\uparrow,\downarrow,\downarrow,\downarrow,\downarrow)$, which by that we mean the $\uparrow$ indices are in $[(C/2)+1:C]$
and $\downarrow$ ones are in $[1:C/2]$. Suppose the choices of $x_{<p}$ and $y_{<p}$ already ``bring'' the walks $W(x_{<p} \circ z)$ and $\barW(y_{<p} \circ z)$ to 
states $2$ and $10$, respectively. From there, $W(x_{<p} \circ z)$, increases the state by two and then decreases it by four, whereas by reflection, $\barW(y_{<p} \circ z)$ does the exact opposite; thus, 
$W(x_{<p} \circ z)$ reaches state $0$ iff $\barW(y_{<p} \circ z)$ reaches $12 (=2m)$. As such, at most one pair of these vertices can be matched in $\bG$ (only vertices whose walk reaches $0$ are matched regardless of being in $\bL$ or $\bR$).}
\label{fig:W}
\end{figure}

\begin{claim}[``Ban constraint'']\label{clm:bG-ban}
	The edge-set $\bE$ respects the ban constraint. 
\end{claim}
\begin{proof}

	Consider an edge $\be \in \bEE{p}$ for  $p=t+2$ for some arbitrary $t \in [L:D]$ and suppose $\be$ is between some $\bL(x) \in \bL_{j,t}$ and $\bR(y) \in \bR_{2m-j,t}$ for $x,y \in [C]^P$. Recall all edges of $\bG$ are of this form. 
	
	By~\Cref{def:proper-blueprint}, the ban constraint requires that for all $z \in [C]^{P-p+1}$, at most one of $\bL(x_{<p} \circ z)$ and $\bR(y_{<p} \circ z)$ is matched in $\bE$. We show that this is the case for every fixed $z \in [C]^{P-p+1}$. 
	\Cref{fig:W} can be helpful with providing an illustration of this step of the proof.

	Consider the walk $W(x_{<p} \circ z)$. Fixing $x_{<p}$ fixes the starting state of the walk (by fixing $x_{\pstart}$), the delay of the walk (by fixing $x_{\pdelay}$) and the states until time $t$ of the walk. 
	By the definition of $\bL(x) \in \bL_{j,t}$, we know that at the time $t$, the walk $W(x_{<p} \circ z)$ will also be at the state $j$. 
	Thus, for $\bL(x_{<p} \circ z)$ to be matched, we need the walk $W(x_{<p} \circ z)$ to reach the state $0$ or equivalently, we need the walk $W_z := w_t(=j),w_{t+1},\cdots,$ to reach $0$ starting from $j$. 
	 Note that the entire movement of the walk $W_z$ is determined by $z$.  
	
	On the other hand, consider the reflection walk $\barW(y_{<p} \circ z)$. Again, fixing $y_{<p}$ and the fact $\bR(y) \in \bR_{2m-j,t}$, in particular means that this walk is at state $2m-j$ at time $t$. 
	Consider the continuation of this walk as $\barW_z := \bar{w}_t (=2m-j), \bar{w}_{t+1},\cdots,$ starting from $2m-j$ which is fully determined by $z$. 
	Given $\barW_z$ is a reflection of $W_z$, if $W_z$ starting from $j$ reaches $0$, then, $\barW_z$ starting from $2m-j$ will reach $2m$. 
	But, that means that the walk $\barW(y_{<p} \circ z)$ ends in the state $2m$ which forces $\bR(y_{<p} \circ z)$ to never be matched in the blueprint. 
	
	Combining the above, we have that between the two vertices $\bL(x_{<p} \circ z)$ and $\bR(y_{<p} \circ z)$, at most one is matched in the blueprint, respecting the ban constraint. 
\end{proof}

We thus established that the blueprint $\bG$ is indeed a proper blueprint. 

\subsection{Value Calculation} 

The last step is to calculate the value of the blueprint which is $\card{\bE}/C^P$. 

\begin{lemma}\label{lem:value}
	The value of blueprint $\bG=\bG_m$ is 
	\[
		value(\bG) \geq \paren{1-\frac{1}{m}} \cdot \frac{2m}{3m+1}-3m \cdot 2^{-m}, 
	\]
\end{lemma}
\begin{proof}
	Given $\bE$ is a matching, and by construction, 
	\begin{align*}
		value(\bG) &= \frac{\card{\bE}}{C^P} \tag{by~\Cref{def:blueprint-value-approx}} \\
		&= \sum_{j=1}^{2m-1} \sum_{t = L}^{D} \frac{\min\paren{\card{\bL_{j,t}},\card{\bR_{2m-j,t}}}}{C^P} \tag{by the definition of edges} \\
		&\geq \sum_{j=1}^{2m-1} \sum_{t = L}^{D} \paren{\frac{\card{\bL_{j,t}}}{C^P} - 2^{-m} \cdot \frac{1}{D}} \tag{by~\Cref{lem:bjt-reflect}} \\
		&\geq \paren{\sum_{j=1}^{2m-1} \sum_{t = L}^{D} \frac{\card{\bL_{j,t}}}{C^P}} - 2m \cdot 2^{-m}. 
	\end{align*}
	We can now bound the inner sum separately as follows:
	\begin{align*}
		\sum_{j=1}^{2m-1} \sum_{t = L}^{D} \frac{\card{\bL_{j,t}}}{C^P} &= \sum_{j=1}^{2m-1} \sum_{t = L}^{D} \Pr\Paren{\text{$W$ $L$-restricted, $W$ reaches $0$, $\maxwalk{W}=j$, $\timemax{W} = t$}} \tag{by the definition of $\bL_{j,t}$} \\
		&= \sum_{j=1}^{2m-1} \Pr\Paren{\text{$W$ $L$-restricted, $W$ reaches $0$, $\maxwalk{W}=j$, $\timemax{W} \in [L:D]$}} \tag{as the events in the sum are all disjoint} \\
		&\geq \sum_{j=1}^{2m-1} \Pr\Paren{\text{$W$ $L$-restricted, $W$ reaches $0$, $\maxwalk{W}=j$, $d(x) \in [L:D-L]$}} \tag{since $W$ is $L$-restricted, having delay in $[L:D-L]$ ensures $\timemax{W} \in [L:D]$} \\
		&\geq \frac{D-2L}{D} \cdot \sum_{j=1}^{2m-1} \Pr\Paren{\text{$W$ $L$-restricted, $W$ reaches $0$, $\maxwalk{W}=j$}} \tag{as $d(x)$ is uniform independently on $[D]$ and $\card{[L:D-L]} = D-2L+1 > D-2L$} \\
		&= \frac{D-2L}{D} \cdot \Pr\Paren{\text{$W$ $L$-restricted, $W$ reaches $0$}} \tag{the sum is over all disjoint possibilities of reaching $0$} \\
		&\geq \frac{D-2L}{D} \cdot  \paren{\Pr\Paren{\text{$\widetilde W$ reaches $0$}}-2^{-m}} \tag{by~\Cref{clm:walk-length} and using $\widetilde W$ as a generic random walk} \\
		&= \frac{D-2L}{D} \cdot \paren{\frac{2m}{3m+1}-2^{-m}} \tag{by~\Cref{lem:start-state}} \\
		&\geq \paren{1-\frac{1}{m}} \cdot \frac{2m}{3m+1}-2^{-m} \tag{as $D = 2m \cdot L$}. 
	\end{align*}
	By plugging in this into the earlier bound, we have, 
	\[
		value(\bG) \geq \paren{1-\frac{1}{m}} \cdot \frac{2m}{3m+1}- 3m \cdot 2^{-m}, 
	\]
	concluding the proof. 
\end{proof}

The lower bound in \Cref{thm:blueprint} now follows immediately from~\Cref{lem:value} and the properness of the blueprint  established in~\Cref{clm:bG-matching,clm:bG-ban}. The upper bound is simply because 
\[
	value(\bG_m) \leq \Pr_{W \sim \dist_m}\paren{\text{$W$ reaches $0$}} = \frac{2m}{3m+1}, 
\]
by~\Cref{lem:start-state}, as every matched vertex in $\bL$ belongs to some $\bL_{j,t}$ and hence its walk has reached $0$. 

This concludes the proof of~\Cref{thm:blueprint}.

\section*{Acknowledgement} 

We are thankful to Aaron Bernstein, Gunadi Gani, Lap Chi Lau, Euiwoong Lee, and Thatchaphol Saranurak for helpful discussions on blueprints, and to Lap Chi Lau and Janani Sundaresan for 
many helpful comments on the presentation of an earlier version of the paper. We are also grateful to Trever Brown for providing us with computational resources for generating
blueprints efficiently in the early stages of this project and to David Woodruff for running multiple LLM prompts on our behalf on Google's Gemini-based models. 

  Sepehr Assadi is additionally grateful to Soheil Behnezhad, Aaron Bernstein, Niv Buchbinder, Deeparnab Chakrabarty, Michael Kapralov, Sanjeev Khanna, Christian Konrad, Lap Chi Lau, Thatchaphol Saranurak, and David Wajc for various discussions
  on semi-streaming matching over the years. 

\section*{AI Acknowledgement} 
We have used Google's Gemini-based models, GPT 5.5 Pro, GPT 5.6 Sol, Claude Opus, and Claude Fable as part of our research. 
Gemini-based models were used briefly during the preliminary stages for a small number of prompts using only the definitions of blueprints, without additional project context, 
and their outputs have not been used in the rest of the project. 
The remaining systems were routinely supplied with the authors' evolving research notes and ideas, and 
have been used for brainstorming and technical assistance. Specifically, the authors provided an optimal solution to a \emph{relaxation} of blueprints (a generalization of the ``density''-argument in~\Cref{app:2-3-upper}),
and a combination of Claude Fable and GPT 5.6 Sol gave the idea for using random walks to take this solution from the relaxations 
to blueprints. These random walks formed the crux of our blueprints. 
GPT 5.6 Sol also assisted with the proof of~\Cref{lem:start-state} by solving the optimization problem after authors' directive to formulate it as a linear program. The proof in the paper 
is different than the AI-generated proof and is our own as we could simply use the provided variable values and directly prove they satisfy the constraints and the bound on the objective. 

AI tools were also used to assist in proofreading and creating figures from given drawings. No AI-generated text appears in the paper: the authors formulated and wrote all the final mathematical content themselves, including all the statements and proofs.
The authors take full responsibility for all content.

\bibliographystyle{halpha-abbrv}
\bibliography{new}

\clearpage
\appendix

\part*{Appendix}

% !TeX root = main.tex 
%!TEX root = main.tex

\section{\emph{Upper Bounding} the Value of Blueprints}\label{app:2-3-upper}

For completeness and to provide further intuition about blueprints, we prove an optimal upper bound for~\Cref{eq:sup}. As stated earlier, this bound also follows from~\Cref{thm:framework} of~\cite{AssadiJX26} and the greedy algorithm for the semi-streaming
matching problem. 

\begin{lemma}\label{thm:value-2/3}
	$\sup\set{value(\bG):\text{$\bG$ a proper blueprint}} \leq \dfrac23$. 
\end{lemma}

\begin{proof}
Fix any blueprint $\bG=(\bL,\bR,\bE)$ with parameters $(P,C)$. Define $\bX_L$ (resp. $\bX_R$) to be the matched vertices in $\bL$ (resp. $\bR$). 
We use $\nu := value(\bG)$ for the simplicity of notation. We will prove that $\nu \leq 2/3$ (throughout we use the fact that $\nu > 1/2$ to simplify some definitions). 

Consider $\bL$ first. For any vertex $\bL(x)$ for $x \in [C]^P$, define 
\[
	D(\bL(x)) := \min_{p \in [P+1]} \quad \Pr_{z \in_R [C]^P}\paren{\bL(z) \in \bX_L \mid z_{<p}=x_{<p}}. 
\]
For each $p \in [P+1]$, the corresponding min-term in the RHS measures the \emph{density} of $\bX_L$ inside vertices consistent with $x_{<p}$. Specifically, for $p=1$, that term is simply $\Pr(\bL(z) \in \bX_L) = \nu$ and for 
$p=P+1$ it is either $1$ or $0$ depending on whether $\bL(x) \in \bX_L$ or not. 

We define this analogously for $\bR$ by replacing $\bX_L$ with $\bX_R$ and considering $\bR(x)$ for $x \in [C]^P$.

\begin{claim}\label{clm:1}
	For any edge $(\bL(x),\bR(y)) \in \bE$, $D(\bL(x)) + D(\bR(y)) \leq 1$. 
\end{claim}
\begin{proof}
	Fix an edge $(\bL(x),\bR(y)) \in \bEE{p}$ for some $p \in [P]$. By the ban constraint (\Cref{def:proper-blueprint}), for each $z \in [C]^{P-p+1}$, at most one of vertices
	$\bL(x_{<p} \circ z)$ and $\bR(y_{<p} \circ z)$ can be matched in $\bE$. Thus, 
	\[
		\Pr_{w \in_R [C]^P}\paren{\bL(w) \in \bX_L \mid w_{<p}=x_{<p}} + \Pr_{w \in_R [C]^P}\paren{\bR(w) \in \bX_R \mid w_{<p}=y_{<p}} \leq 1. 
	\]
	Taking the minimum on each term can only decrease the LHS, implying the claim. \Qed{clm:1} 
	
\end{proof}

By~\Cref{clm:1}, we have
\[
	\Pr_{(\bL(x),\bR(y)) \in_R \bE}\paren{D(\bL(x)) > \frac12 ~\wedge~ D(\bR(y)) > \frac12} = 0, 
\]
which, by union bound and since $\bE$ is a matching, implies that 
\begin{align}
	1 &= \Pr_{(\bL(x),\bR(y)) \in_R \bE}\paren{D(\bL(x)) \leq \frac12 \vee D(\bR(y)) \leq \frac12} \notag \\
	&\leq \Pr_{\bL(x) \in_R \bX_L}\paren{D(\bL(x)) \leq \frac12} + \Pr_{\bR(y) \in_R \bX_R}\paren{D(\bR(y)) \leq \frac12}.\label{eq:one-side} 
\end{align}

\begin{claim}\label{clm:2}
	$
		\Pr_{\bL(x) \in_R \bX_L}\paren{D(\bL(x)) \leq \dfrac12} \leq \dfrac{1-\nu}{\nu}. 
	$
\end{claim}
\begin{proof}
	We say $x \in [C]^P$ is \textbf{bad} if $D(\bL(x)) \leq 1/2$. We further say $x$ is $\bm{p}$-\textbf{bad} for $p \in [2:P+1]$ iff $p$ is the \emph{smallest} index where
	\[
		\Pr_{z \in_R [C]^P}\paren{\bL(z) \in \bX_L \mid z_{<p}=x_{<p}} \leq \frac 12. 
	\]
	Any bad vertex is $p$-bad for a unique choice of $p \in [2:P+1]$. We have 
	\begin{align}
		\Pr_{\bL(x) \in_R \bX_L}\paren{D(\bL(x)) \leq \frac12} &= \Pr_{x \in_R [C]^P}\paren{\text{$x$ is bad} \mid \bL(x) \in \bX_L} \notag \\
		&= \frac{\Pr_{x \in_R [C]^P}\paren{\text{$x$ is bad} \wedge \bL(x) \in \bX_L}}{\Pr_{x \in_R [C]^P}\paren{\bL(x) \in \bX_L}}. \label{eq:clm2}
	\end{align}
	The denominator in the RHS is $\nu$ by definition. We now bound the numerator. 
	\begin{align*}
		\Pr_{x \in_R [C]^P}\paren{\text{$x$ is bad} \wedge \bL(x) \in \bX_L} &= \sum_{p=2}^{P+1} \Pr_{x \in_R [C]^P}\paren{\text{$x$ is $p$-bad} \wedge \bL(x) \in \bX_L} \tag{as a bad $x$ is $p$-bad for a unique $p$} \\
		&= \sum_{p=2}^{P+1} \Pr_{x \in_R [C]^P}\paren{\text{$x$ is $p$-bad}} \cdot \Pr_{x \in_R [C]^P}\paren{\bL(x) \in \bX_L \mid \text{$x$ is $p$-bad}} \\
		&\leq \sum_{p=2}^{P+1} \Pr_{x \in_R [C]^P}\paren{\text{$x$ is $p$-bad}} \cdot \frac12 \tag{explained below} \\
		&= \frac{1}{2} \Pr_{x \in_R [C]^P}\paren{\text{$x$ is bad}} \tag{again, since $p$-bad events partition the bad event}. 
	\end{align*}
	The inequality above holds because conditioning on $x$ being $p$-bad requires fixing $x_{<p}$, and so by sampling the remaining coordinates and the definition of $p$-bad, we know $\bL(x)$ belongs to $\bX_L$ with probability at most $1/2$. 
	
	Finally, we can upper bound the probability that $x$ is bad (rather crudely) as follows: 
	\begin{align*}
		\nu &= \Pr_{x \in_R [C]^P}\paren{\bL(x) \in \bX_L} = \Pr_{x \in_R [C]^P}\paren{\bL(x) \in \bX_L \wedge \text{$x$ is bad}} + \Pr_{x \in_R [C]^P}\paren{\bL(x) \in \bX_L \wedge \text{$x$ is not bad}} \\
		&\leq \frac12 \cdot \Pr_{x \in_R [C]^P}\paren{\text{$x$ is bad}} + \Pr_{x \in_R [C]^P}\paren{\text{$x$ is not bad}} \tag{by the previous equation} \\
		&\leq \frac12 \cdot \Pr_{x \in_R [C]^P}\paren{\text{$x$ is bad}} + 1 - \Pr_{x \in_R [C]^P}\paren{\text{$x$ is bad}}.
	\end{align*}
	Therefore, we have, 
	\[
		\Pr_{x \in_R [C]^P}\paren{\text{$x$ is bad}} \leq 2 \cdot \paren{1-\nu}. 
	\]
	Plugging these bounds in~\Cref{eq:clm2} implies that 
	\[
		\Pr_{\bL(x) \in_R \bX_L}\paren{D(\bL(x)) \leq \frac12} \leq \frac{1-\nu}{\nu}, 
	\]
	concluding the proof. \Qed{clm:2}
	
\end{proof}

The same exact proof as~\Cref{clm:2} also holds for vertices in $\bR$ and $\bX_R$ by symmetry. Plugging in these bounds in~\Cref{eq:one-side}, we have, 
\[
	1 \leq \Pr_{\bL(x) \in_R \bX_L}\paren{D(\bL(x)) \leq \frac12} + \Pr_{\bR(y) \in_R \bX_R}\paren{D(\bR(y)) \leq \frac12} \leq 2 \cdot \frac{1-\nu}{\nu},
\]
which in turn implies that $\nu \leq  2/3$, concluding the proof of~\Cref{thm:value-2/3}. 
\end{proof}

\end{document}